\documentclass[a4paper, authorcolumns, nolineno]{lipics-v2021}
\usepackage{mathtools}
\usepackage{cite}
\usepackage{float}
\usepackage{graphicx}
\usepackage[ruled]{algorithm2e}

\graphicspath{ {./images/} }

\newcommand{\arcsinh}{\operatorname{arcsinh}}

\newcommand{\RE}{\mathbb{R}}  
\newcommand{\HY}{\mathbb{H}}

\newcommand{\eps}{\varepsilon}          
                 
\newcommand{\bd}{\partial}
\newcommand{\bdOmega}{\partial \kern+1pt \Omega} 

\newcommand{\SP}{\kern+1pt}             
\newcommand{\revFunk}[1]{{}^r \kern-1pt F_{#1}}

\title{The Hyperbolic Surface Distance, Diameter, and Dirichlet Problems}
\titlerunning{\small The Hyperbolic Surface Distance, Diameter, and Dirichlet Problems}

\author{Vincent Despré}{Institut universitaire de France (IUF), Université de Lorraine, CNRS, LORIA, Nancy, France}{vincent.despre@loria.fr}{}{}
\author{Auguste H. Gezalyan}{Université de Lorraine, CNRS, Inria, LORIA, F-54000 Nancy, France}{octavo@umd.edu}{https://orcid.org/0000-0002-5704-312X}{}
\author{Marc Pouget}{Université de Lorraine, CNRS, Inria, LORIA, F-54000 Nancy, France}{marc.pouget@loria.fr}{https://orcid.org/0000-0001-8085-4134}{}

\authorrunning{\small Despr\'e, Gezalyan, Pouget}

\Copyright{Despr\'e, Gezalyan, Pouget}

\ccsdesc[500]{Theory of computation~Computational geometry}
\keywords{Hyperbolic surfaces, Poincare disk, Voronoi diagrams}

\date{\today}

\EventEditors{}
\EventLongTitle{}
\EventShortTitle{}
\EventAcronym{}
\EventYear{}
\EventDate{}
\EventLocation{}
\EventLogo{}
\SeriesVolume{}
\ArticleNo{}

\nolinenumbers
\funding{This work was partially supported by grant ANR-23-CE48-0017 of the French National Research Agency ANR (project Abysm) and partially supported by grant ANR-25-CE40-0416 of the French National Research Agency (project SUGAR).}
\begin{document}

\maketitle

\begin{abstract}
Despite the prominence of hyperbolic surfaces in mathematics, basic algorithmic questions about them, even computing the distance between two points, have remained open, leaving many features of these surfaces inaccessible. The classical machinery assumes a polyhedral structure absent on a smooth surface. We remove these obstacles. We begin with an efficient $O(g^2)$ algorithm for the distance between two points, where $g$ is the genus of the surface. Building on it, we obtain an $O(g^2 \log g)$ method for answering distance queries from a fixed source and, as a consequence, for recentering a Dirichlet domain around an arbitrary point. This understanding of distances on the surface then lets us approximate the diameter to within any $\eps$ in time $O(g^3 \log g / \eps^2)$. We further show that the diameter, a single real number encoding a great deal about the surface, is exactly computable. Its hyperbolic cosine is an algebraic number over the field encoding the coefficients of the hyperbolic isometries defining the surface.

\end{abstract}

\section{Introduction}

Hyperbolic geometry, the geometry of constant negative curvature, underlies structures as varied as complex networks~\cite{krioukov2010hyperbolic}, phylogenetic trees~\cite{matsumoto2021novel}, and hierarchical embeddings in machine learning~\cite{nickel2017poincare}. It also plays a central role in theoretical physics through anti-de Sitter spacetimes~\cite{Carroll2019}. Its importance is even more fundamental in mathematics: by the Uniformization Theorem~\cite{abikoff1981uniformization}, every closed surface of genus at least two admits a unique hyperbolic metric in its conformal class. Consequently, hyperbolic surfaces are not exceptional objects but rather the canonical geometric representatives of surfaces of higher genus. Despite their ubiquity, fundamental algorithmic questions concerning hyperbolic surfaces are open.

Recent years have seen the development of several algorithmic tools for hyperbolic surfaces. Delaunay triangulations have been studied in increasing generality, from low-genus examples to arbitrary hyperbolic surfaces~\cite{bogdanov2016delaunay,iordanov2016implementing,despreFlippingGeometricTriangulations2020}. More recently, an algorithm for computing the Dirichlet domain of a hyperbolic surface from an arbitrary fundamental polygon~\cite{despreComputingDirichletDomain2023} provided a canonical geometric representation that is independent of the quality of the input. Another line of work introduced $\eps$-nets as a geometric discretization suitable for approximation algorithms~\cite{despre2024computing}. While Dirichlet domains provide a global description of the geometry of a surface, $\eps$-nets capture its local structure. In this paper, we show how these two complementary representations lead to efficient algorithms for fundamental geometric problems on hyperbolic surfaces.

Our first contribution concerns the computation of shortest-path distances. Computing distances on hyperbolic surfaces is closely related to the shortest-path problem on polyhedral surfaces. Two classical algorithms play a central role in this context: the algorithm of Mitchell, Mount, and Papadimitriou~\cite{mitchell1987discrete}, which is remarkably robust, and the optimal algorithm of Chen and Han~\cite{chen1996shortest}, whose adaptation to hyperbolic surfaces requires a more delicate analysis. We show how both approaches can be extended to hyperbolic surfaces represented by Dirichlet domains. This yields the following result.

\begin{theorem}
Let $S$ be a hyperbolic surface of genus $g$ represented by a Dirichlet domain $D\subset\mathbb{H}^2$, and let $\alpha,\omega\in S$ be specified by lifts in $D$. Then the distance $d_S(\alpha,\omega)$ can be computed in $O(g^2)$ time.
\end{theorem}

A key contribution of this work is the insight that the wavefront propagated by our shortest-path algorithm implicitly encodes the Dirichlet domain centered at the source point. This leads simultaneously to an efficient single-source shortest-path data structure and to an algorithm for recentering Dirichlet domains. It complements the algorithm of Despr\'e, Kolbe, Parlier, and Teillaud~\cite{despreComputingDirichletDomain2023}, which computes a canonical Dirichlet domain from an arbitrary fundamental polygon but does not allow the center to be prescribed. We obtain the following.

\begin{theorem}
Let $S$ be a hyperbolic surface of genus $g$ represented by a Dirichlet domain $D\subset\mathbb{H}^2$, and let $\alpha\in S$ be specified by a lift $\tilde{\alpha}\in D$. After $O(g^2\log g)$ preprocessing, shortest-path queries from $\alpha$ to arbitrary points of $S$ can be answered in $O(\log g)$ time. The same preprocessing yields the Dirichlet domain $D_{\tilde{\alpha}}$.
\end{theorem}

Our third contribution concerns the diameter of a hyperbolic surface. Unlike the distance queries above, which require only the Dirichlet domain, the diameter is a global invariant and requires a surface-wide discretization. We therefore assume an extended representation: a Dirichlet domain together with the Delaunay triangulation of a pseudo $\eps$-net on the thick part of the surface, as produced by~\cite{despre2024computing}. The diameter is considerably more challenging to compute exactly; to the best of our knowledge, exact values are currently known only for a single hyperbolic surface in each genus~\cite{stepanyantsDiameterCompactRiemann2023}. We show that this discretization is accurate enough to obtain an efficient approximation algorithm.

\begin{theorem}
Let $S$ be a hyperbolic surface given by a Dirichlet domain together with the Delaunay triangulation of a pseudo $\eps/2$-net on the thick part of $S$. Then an additive $\eps$-approximation of the diameter of $S$ can be computed in time $O(g^3 \log g / \eps^2)$.
\end{theorem}

Finally, we investigate the exact computation of the diameter from the viewpoint of arithmetic complexity. Although our approximation algorithm is efficient, it is natural to ask whether the exact diameter can be computed algorithmically. We show that this problem reduces to solving a finite collection of polynomial systems. While the resulting procedure is exponential and therefore not intended as a practical algorithm, it establishes that the diameter is determined by algebraic data over the field generated by the input.

\begin{theorem}
Let $S$ be a hyperbolic surface given by a finite set
of generators of $\Gamma$, each a hyperbolic isometry. Let $\mathbb{K}$ be the rational extension containing the real and imaginary parts of the coefficients of the generators of $\Gamma$. 
The exact diameter of $S$ can be computed by solving exponentially many systems of polynomial equations. In particular, $\cosh(\operatorname{diam}(S))$ lies in a finite extension of $\mathbb{K}$.
\end{theorem}

The remainder of the paper is organized as follows. After a short preliminaries section, Section~\ref{sec:distances} presents our algorithm for computing the shortest path between two points. Section~\ref{sec:datastructure} extends this construction to obtain a single-source shortest-path data structure and an algorithm for computing Dirichlet domains centered at arbitrary points. Section~\ref{sec:diameter} is devoted to the approximation of the diameter using $\eps$-nets. Finally, Section~\ref{sec:exact-diameter} investigates the arithmetic complexity of the exact diameter and proves that it can be computed by solving a finite collection of polynomial systems.

\section{Preliminaries}
\subsection{Hyperbolic Surfaces}
The hyperbolic plane is the unique $2$-dimensional simply connected Riemannian manifold with constant curvature $-1$. While it cannot be isometrically embedded in Euclidean space, several models allow us to study it. Throughout, we work with the \emph{Poincar\'e disk model}. 

\begin{definition}[Poincar\'e disk]
The Poincar\'e disk, $\mathbb{H}^2$, is the open unit disk in the complex plane, equipped with the metric $d_{\HY^2}: \HY^2 \times \HY^2 \rightarrow \RE^{\geq 0}$, such that:
\[
d_{\mathbb{H}^2}(z,w) = \operatorname{arccosh}\!\left(1 + 
\frac{2|z-w|^2}{(1-|z|^2)(1-|w|^2)}\right).
\]
\end{definition}

Note that the expression inside of $\operatorname{arccosh}$ is a rational function of $z$ and $w$. Throughout the paper:

\begin{definition}[Hyperbolic surface]
 A hyperbolic surface is a closed orientable Riemannian $2$-manifold without boundary and of constant curvature $-1$.
\end{definition}

Every such surface $S$ has genus $g \geq 2$ by the Gauss-Bonnet theorem and is isometric to a quotient $S=\mathbb{H}^2 / \Gamma$, each $\gamma \in \Gamma$ is a hyperbolic isometry:
\[
\gamma(z) = \frac{uz+v}{\bar{v}z+\bar{u}}, \qquad |u|^2 - |v|^2 = 1.
\]

A \textbf{geodesic} is a curve that is locally distance-minimizing. On a negatively curved surface, every homotopy class of paths with fixed endpoints contains a unique geodesic. We call any globally shortest path between two points a \textbf{distance path}. In the hyperbolic plane $\mathbb{H}^2$, every geodesic is a distance path. This is not true on a hyperbolic surface $S$. A curve on $S$ is a geodesic if and only if each of its lifts to $\mathbb{H}^2$ is a geodesic, but not all of these are distance paths on $S$. The length of a distance path between two points on $S$ is given by $d_S(s_1,s_2)=\min_{\gamma\in\Gamma} d_{\mathbb{H}^2}(\tilde s_1,\gamma\tilde s_2),$ where $\tilde s_1,\tilde s_2\in\mathbb{H}^2$ are arbitrary lifts of $s_1,s_2\in S$.

\subsection{The discrete geodesic problem}\label{sec:discrete-geodesic}

The \emph{discrete geodesic problem} asks how to compute the shortest path between two points on the surface of a polyhedron $P$. Mitchell, Mount, and Papadimitriou \cite{mitchell1987discrete} introduce a continuous Dijkstra algorithm for solving this problem which uses a technique called \emph{windowing}. A \emph{window} is a sub-interval of an edge of a triangulation $P$, together with a point representing a translated source, called a \emph{virtual start}; the distance from the source to any point in the window is measured as the distance from this virtual start. Windows are propagated across faces of the triangulation of $P$ by a continuous Dijkstra sweep. An interval is split whenever a competing image dominates part of it. On a polyhedron with $n$ edges this produces $O(n^2)$ windows and yields an $O(n^2 \log n)$ algorithm, along with a shortest path map supporting $O(\log n)$ point queries by point location in the resulting subdivision.

Chen and Han \cite{chen1996shortest} removed the logarithmic factor from the complexity of Mitchell, Mount, and Papadimitriou \cite{mitchell1987discrete} using a \emph{sequence tree}. Rather than using a priority queue based approach, they process windows in breadth-first order. The nodes of their tree are windows. When a window crosses into the adjacent face, its interval is projected onto the two opposite edges, and each resulting subinterval becomes a child node. They bound the horizontal growth of the tree using the \emph{one-angle-one-split} property: at each vertex $A$ of the triangulation of $P$, among all images competing for the angle at $A$, only the one closest to $A$ may fork into children on both incident edges, while the others produce a single child. This keeps the number of leaves at $O(n)$ per level. By proving that they only need to check a depth of $O(n)$ levels,  they get $O(n^2)$ windows/nodes in $O(n^2)$ time.

In our setting we adapt both of these methods. Chen and Han's query structure, however, relies on the planar unfolding of the surface and on curvature being concentrated at vertices \cite{aronov1991nonoverlap}, neither of which holds for a hyperbolic surface. Whether their $O(n^2)$ preprocessing time can be matched in our setting is an open problem.

\subsection{Dirichlet domains.}
Fix a basepoint $\tilde u \in \mathbb{H}^2$. The \emph{Dirichlet domain} 
centered at $\tilde u$ is the Voronoi cell of $\tilde u$ in the orbit 
$\Gamma \tilde u$:
\[
D_{\tilde u} = \bigl\{ \tilde v \in \mathbb{H}^2 \mid 
d_{\mathbb{H}^2}(\tilde u, \tilde v) \leq 
d_{\mathbb{H}^2}(\gamma \tilde u, \tilde v) 
\text{ for all } \gamma \in \Gamma\setminus\{e\} \bigr\}.
\]
The translations $\{\gamma D_{\tilde u}\}_{\gamma \in \Gamma}$ tile $\mathbb{H}^2$: their union is all of $\mathbb{H}^2$ and their interiors are pairwise disjoint. The interior of $D_{\tilde u}$ contains exactly one lift of each point of $S$ not lying on the projection of $\partial D_{\tilde u}$; the remaining points have their lifts on $\partial D_{\tilde u}$, identified by the side-pairings. The boundary of $D_{\tilde u}$ consists of $k$ geodesic segments called \emph{sides}, where $4g \leq k \leq 12g-6$ \cite{despreComputingDirichletDomain2023}, grouped into $k/2$ identified pairs called \emph{portals}. Each pair is related by a \emph{side-pairing} isometry $\gamma \in \Gamma$; these generate $\Gamma$ and determine $S$ completely.

\begin{definition}[Portal]
Let $s_1, \ldots, s_k$ be the sides of a Dirichlet domain $D_{\tilde u}$. The sides are grouped into $k/2$ pairs $\{p, p'\}$, where each pair is related by a unique \emph{side-pairing} isometry $\gamma_p \in \Gamma$ with $\gamma_p(p') = p$. Each such pair is called a \emph{portal} of $D_{\tilde u}$.
\end{definition}

\section{Distance path between two points}\label{sec:distances}

\subsection{Supporting lemmas and proofs}
\label{sec:ShortesPath}
To compute the distance path between two points on $S$ we adapt the sequence tree approach of Chen \& Han \cite{chen1996shortest} based on the windowing shortest-path framework of the discrete geodesic problem \cite{mitchell1987discrete} (see Sec \ref{sec:discrete-geodesic}). We exploit the convexity of the Dirichlet domain $D_{\tilde c}$ for a fixed base-point $\tilde c \in \HY^2$ and the distance minimizing properties of its edges \cite{despre2024representing}. The boundary of the Dirichlet domain $D_{\tilde c}$ consists of $k$ geodesic edges, where $4g \leq k \leq 12g-6$, forming $\frac{k}{2}$ identified portal pairs, each identified by a hyperbolic isometry $\gamma_p$. We construct the \emph{fan triangulation} of $D_{\tilde c}$ by connecting $\tilde c$ to each vertex of $\partial D_{\tilde c}$ by a geodesic in $\mathbb{H}^2$, called \emph{fan edges} (see Fig~\ref{fig:setup} (c)). This produces $O(g)$ triangles whose edges we use to propagate distance information from $\tilde{\alpha}$ to $\tilde{\omega}$. We emphasize that every computation will be done in $\mathbb{H}^2$ and not directly on $S$. Indeed, given start and end points $\alpha, \omega \in S$ we will use their lifts $\tilde{\alpha}, \tilde{\omega} \in D_{\tilde c}$, (see Fig~\ref{fig:setup} (a)-(b)), and then compute windows with virtual starts being points of $\Gamma\tilde{\alpha}$.

\begin{figure}[h]
\centering
\includegraphics[page = 5, scale=.7]{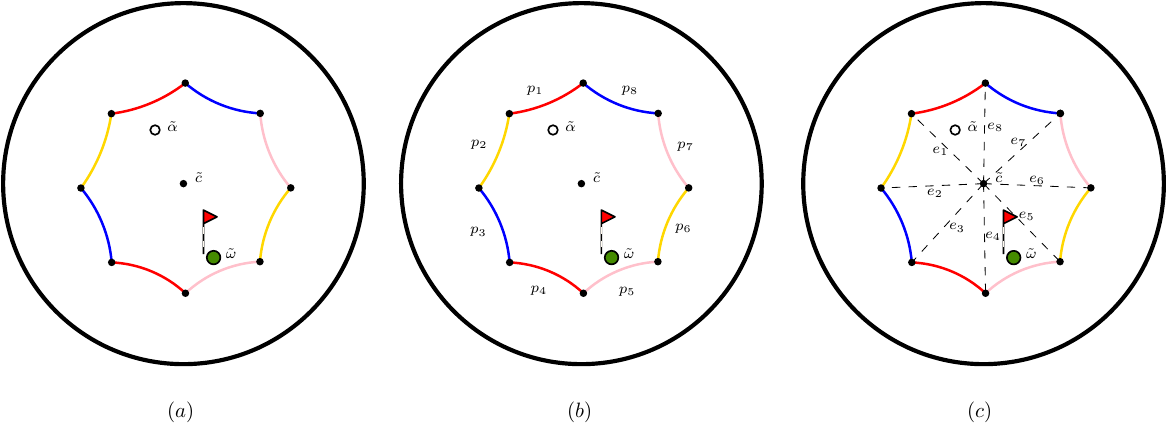}
  \caption{(a) A Dirichlet domain with base point $\tilde c$, portals identified by colors (pairs of edges), an end $\tilde{\omega}$ and a start $\tilde{\alpha}$, (b) with portals $p_i$ enumerated, (c) and fan edges $e_i$ around $\tilde c$. }
  \label{fig:setup}
\end{figure}

\begin{definition}[Virtual start]\label{def:virtual-start-distance}
Let $p_1, \dots, p_j$ be a sequence of portals crossed in order. The \emph{virtual start} associated with this sequence is $\tilde{\alpha}' = \gamma_{p_j}^{-1} \circ \cdots \circ \gamma_{p_1}^{-1}(\tilde{\alpha})$.
\end{definition}

This takes on the role of unfolding images in Chen \& Han \cite{chen1996shortest}. Rather than unfolding triangles in the Euclidean plane, each portal crossing applies a hyperbolic isometry $\gamma_{p_i}^{-1}$, placing a virtual start in a different fundamental domain of $S$ in $\HY^2$. This is equivalent to placing the start in the neighboring tile across the portal, so that the geodesic from this relocated start reaches its target exactly as a path through the portal would.

\begin{remark}[Virtual start distance]
\label{rem:virtual-start-distance}
For any point $\tilde{q}$ in a Dirichlet domain $D_{\tilde c}$ and virtual start $\tilde{\alpha}'$ associated with portal sequence $p_1, \dots, p_j$, the length of the shortest path from $\tilde{\alpha}$ to $\tilde{q}$ crossing portals $p_1, \dots, p_j$ in order equals $d_{\mathbb{H}^2}(\tilde{\alpha}', \tilde{q})$.
\end{remark}

Indeed, portal $p_1$ is a geodesic edge of $\partial D_{\tilde c}$ identified with its partner $p_1'$ via $\gamma_{p_1}$. Crossing $p_1$ from $\tilde{\alpha}$ is equivalent to applying $\gamma_{p_1}^{-1}$ to $\tilde{\alpha}$ and placing the virtual start at $\gamma_{p_1}^{-1}(\tilde{\alpha})$ in the cell that borders $D_{\tilde c}$ along $p_1'$. Since $\gamma_{p_1}^{-1}$ is a hyperbolic isometry the shortest path from $\tilde{\alpha}$ to $\tilde{q}$ crossing $p_1$ has length $d_{\mathbb{H}^2}(\gamma_{p_1}^{-1}(\tilde{\alpha}), \tilde{q})$. Applying this argument successively to each portal $p_2, \dots, p_j$ gives $\tilde{\alpha}' = \gamma_{p_j}^{-1} \circ \cdots \circ \gamma_{p_1}^{-1}(\tilde{\alpha})$ and total length $d_{\mathbb{H}^2}(\tilde{\alpha}', \tilde{q})$.

To establish the correctness and complexity of our algorithm, we need to bound the number of portals traversed by a distance path. The key observation is that two distance paths cannot cross more than once. We prove the following two lemmas, which are essentially tailored versions of results from~\cite{despre2024representing}.

\begin{lemma}[Fan edges]\label{lem:fan-happy}
    A distance path on $S$ crosses each fan edge at most once.
\end{lemma}
\begin{proof}
    Each fan edge is a distance path, by Lemma~3 of \cite{despre2024representing}, and hence intersects another distance path at most once by Lemma~1 of \cite{despre2024representing}.
\end{proof}

\begin{lemma}[Portal sides]\label{lem:portal-happy}
    A distance path on $S$ crosses each portal side at most twice.
\end{lemma}
\begin{proof}
    By Proposition~4 of \cite{despre2024representing}, each portal side is either a distance path or a half-minimizer. In the first case, two distance paths intersect at most once by Lemma~1 of \cite{despre2024representing}. In the second case, a distance path intersects a half-minimizer (the concatenation of at most two distance paths) at most twice by Lemma~2 of \cite{despre2024representing}.
\end{proof}

\begin{lemma}\label{lem:connected}
    Let $\tilde{\alpha}_1$ and $\tilde{\alpha}_2$ be two virtual starts. The region $\{\tilde q\in D_{\tilde c}\mid d_{\mathbb{H}^2}(\tilde q,\tilde{\alpha}_1) < d_{\mathbb{H}^2}(\tilde q,\tilde{\alpha}_2) \}$ is a geodesically convex subset of $D_{\tilde c}$.
\end{lemma}
\begin{proof}
The bisector between $\tilde{\alpha}_1$ and $\tilde{\alpha}_2$ in $\mathbb{H}^2$ is a geodesic $J$. The region $\{\tilde q \in \mathbb{H}^2 \mid d_{\mathbb{H}^2}(\tilde q, \tilde{\alpha}_1) < d_{\mathbb{H}^2}(\tilde q, \tilde{\alpha}_2)\}$ is a half-space of $\mathbb{H}^2$ bounded by $J$, which is geodesically convex since half-spaces in $\mathbb{H}^2$ are geodesically convex. The intersection of this half-space with $D_{\tilde c}$ is geodesically convex, since $D_{\tilde c}$ is geodesically convex and the intersection of two geodesically convex sets in $\mathbb{H}^2$ is geodesically convex. In particular, its intersection with any edge of the triangulation is a connected interval.
\end{proof}

\begin{algorithm}[!ht]
  \caption{Shortest distance} 
  \label{alg:sd}

  \KwIn{$\alpha$ and $\omega$ two points of  a hyperbolic surface $S$ given by a Dirichlet domain $D_{\tilde c}$ with $k$ sides and $\tilde \alpha$ and $\tilde \omega$ their lifts in $D_{\tilde c}$.}

  \KwOut{$d_S(\alpha, \omega)$}

    \textbf{Initialize.} Create a sequence tree with root $\tilde{\alpha}$. For each edge of the triangle containing $\tilde{\alpha}$ create a  window with virtual start $\tilde{\alpha}$ covering its entire length and insert it as a child of the root. If $\tilde{\omega}$ lies in this triangle, record $d_{\mathbb{H}^2}(\tilde{\alpha}, \tilde{\omega})$ as the current best distance. Create the angle occupant table of the vertices, set the angles of the starting triangle to the root of the sequence tree and the rest to empty.

\For{$i$ from 1  to $3k$}{
   For each leaf of the $i^{th}$ level of the sequence tree, its window is given by a virtual start $\tilde{\alpha}'$, an edge $e$ of the fan triangulation of $D_{\tilde c}$ and an interval on $e$.  
   
   Insert as children of the current leaf the windows determined as follows:
   
    \textbf{\quad (a) Update the virtual start.} Let $\tilde{\alpha}'$ be the virtual start and $e$ the edge of the leaf's window. If $e$ is a portal side $p_j$, crossing it gives the new virtual start $\tilde{\alpha}'' = \gamma_{p_j}^{-1}(\tilde{\alpha}')$ by Definition~\ref{def:virtual-start-distance} (See Fig \ref{fig:portal} (a)). Else $e$ is a fan edge, the virtual start remains the same for $\tilde{\alpha}'' = \tilde{\alpha}'$. 

    \textbf{\quad (b) Identify the next triangle.} If $e$ is a portal side $p_i$, the window propagates into the triangle of $D_{\tilde c}$ incident to the partner side $p_i'$ (See Fig \ref{fig:portal} (b)-(c)). If $e$ is a fan edge, it is shared by exactly two triangles of the fan triangulation, and the window propagates into the one not containing the window's parent.
        
    \textbf{\quad (c) Propagate the window.} Let $A$ be the vertex of the triangle opposite $e$, and let $e_1$, $e_2$ be the two sides of the triangle meeting at $A$. 
    
    If the geodesic cone (called shadow in \cite{chen1996shortest}) from $\tilde{\alpha}''$ through the interval on $e$ does not contain $A$, create one child window on whichever of $e_1$, $e_2$ the cone reaches. In each case, the child's interval is the intersection of the geodesic cone from $\tilde{\alpha}''$ with the corresponding edge. 
    
    Otherwise we compare $d_{\mathbb{H}^2}(\tilde{\alpha}'', A)$ against $d_{\mathbb{H}^2}(\tilde{\alpha}''', A)$, where $\tilde{\alpha}'''$ is the virtual start of the window currently occupying $\angle e_1 A e_2$ (See Fig \ref{fig:portal} (d)).
    
        \textbf{\quad\quad i.} If $\tilde{\alpha}''$ is closer to $A$, its window takes over occupancy and creates two child windows on $e_1$ and $e_2$. 
        The subtree rooted at the child of the previous occupant on the side where $\tilde{\alpha}''$ is closer (See Fig \ref{fig:portal} (e)) is deleted from the sequence tree.
        
        \textbf{\quad\quad ii.}  Otherwise $\tilde{\alpha}''$ creates a single child window, on whichever of $e_1$, $e_2$ it does not lose entirely to $\tilde{\alpha}'''$ (See Fig \ref{fig:portal} (f)).
        
    \textbf{\quad (d) Update the best distance.}  If $\tilde{\omega}$ lies in the next triangle and within the geodesic cone from $\tilde{\alpha}''$, update the current best distance if $d_{\mathbb{H}^2}(\tilde{\alpha}'', \tilde{\omega})$ improves on it.
}
  
    \Return the current best distance.
\end{algorithm}

\subsection{Shortest Distance Algorithm}
The input of Algorithm~\ref{alg:sd} is a Dirichlet domain $D_{\tilde c}$, a start $\tilde{\alpha}$ and an end $\tilde{\omega}$ both inside $D_{\tilde c}$, lifts of $\alpha,\omega \in S$. The algorithm works within the fan triangulation of $D_{\tilde c}$, whose edges are the fan edges from $\tilde{c}$ together with the portal sides. The algorithm tracks windows only on the fan triangulation of $D_{\tilde c}$ itself and never on edges of other lifts.

The algorithm maintains a \emph{sequence tree} whose nodes are windows and an \emph{angle occupant table}, storing, for each vertex $v$ of the fan triangulation, a pointer to the window whose virtual start is currently closest to $v$. This is used to enforce the one-angle-one-split property of \cite{chen1996shortest}: at most one window may create children on both edges adjacent to $v$.

The sequence tree is processed level by level in breadth-first order, following \cite{chen1996shortest}. By Lemma~\ref{lem:fan-happy} and Lemma~\ref{lem:portal-happy}, any shortest path crosses each portal side at most twice and each fan edge at most once. As we show in the next section in Lemma~\ref{lem:correctnessSequence}, we may terminate after reaching $3k$ levels.

\begin{figure}[h]
\centering
\includegraphics[page = 1, scale=.7]{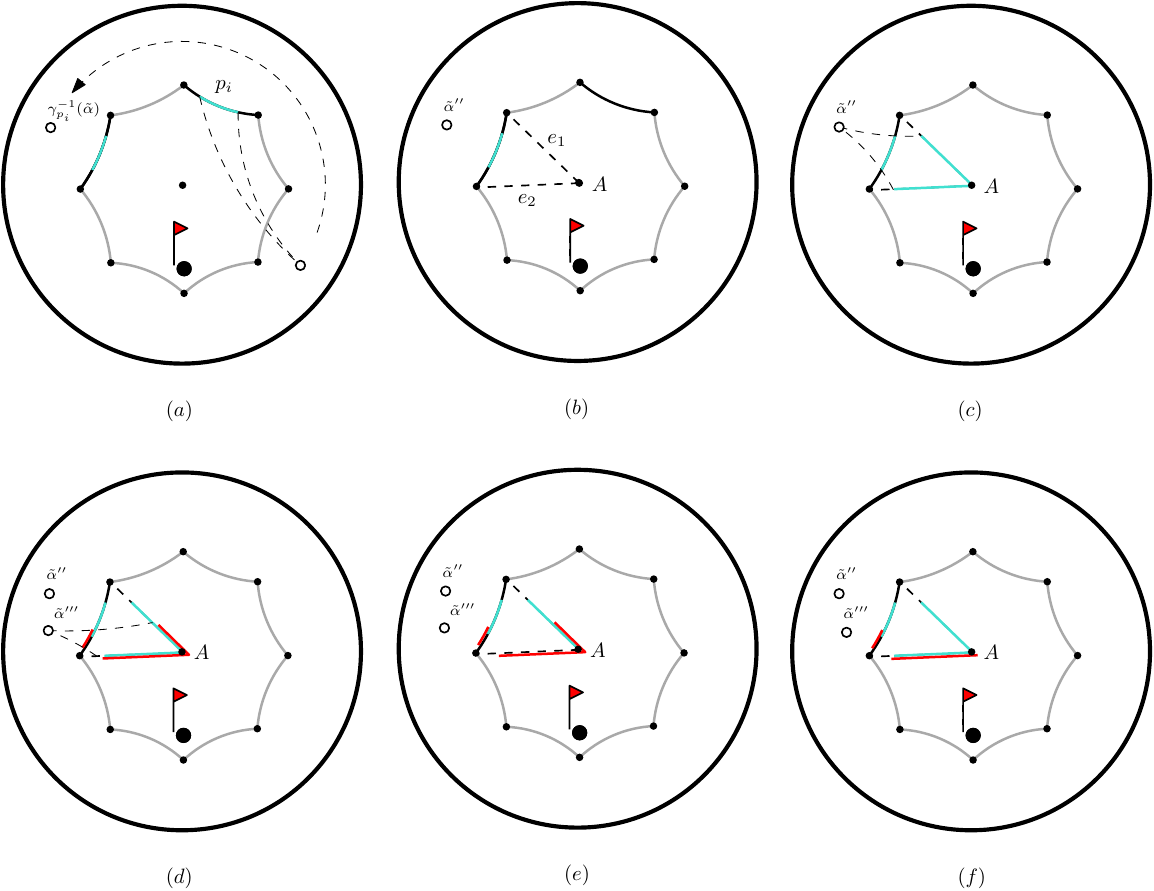}
    \caption{Propagating a window across a portal. (a) Crossing $p_i$ gives the virtual start $\tilde{\alpha}'' = \gamma_{p_i}^{-1}(\tilde{\alpha}')$. (b) The far-side triangle, with vertex $A$ and edges $e_1$, $e_2$. (c) The geodesic cone from $\tilde{\alpha}''$. (d) The cone contains $A$, so $\tilde{\alpha}''$ is compared against the occupant $\tilde{\alpha}'''$. (e) $\tilde{\alpha}'''$ is closer to $A$ and takes over, creating two children. (f) $\tilde{\alpha}''$ is farther and creates one.}
  \label{fig:portal}
\end{figure}

\subsection{Correctness and Complexity}
We prove that Algorithm~\ref{alg:sd} returns the correct shortest distance from $\alpha$ to $\omega$ on $S$, following the correctness argument of \cite{chen1996shortest}.

\begin{lemma}[Correctness]
\label{lem:correctnessSequence}
    Algorithm~\ref{alg:sd} terminates after $3k$ levels and returns distance $d^*= d_S(\alpha, \omega)$.
\end{lemma}

\begin{proof}
We fix lifts $\tilde{\alpha}, \tilde{\omega} \in D_{\tilde c}$. Since $D_{\tilde c}$ is a fundamental domain for $S$, any path on $S$ from $\alpha$ to $\omega$ lifts to a path in $\mathbb{H}^2$ from $\tilde{\alpha}$ to some element of the orbit $\Gamma\tilde{\omega}$, crossing portal sides each time it leaves $D_{\tilde c}$. Hence paths on $S$ correspond exactly to sequences of portal crossings in $D_{\tilde c}$. We show $d^* \geq d_S(\alpha,\omega)$ and $d^* \leq d_S(\alpha,\omega)$.

For the first inequality, since $d^*$ is the length of a valid path on $S$ from $\alpha$ to $\omega$ crossing a sequence of portals and fan edges, it must be that $d^* \geq d_S(\alpha,\omega)$.

For the second inequality, let $\pi$ be a distance path from $\alpha$ to $\omega$ on $S$. The algorithm propagates windows through every edge and portal except for self-propagation: a window arriving through a portal $p_i$ is not propagated back through $p_i$ itself. This never occurs on a shortest path because the path from a virtual start through the triangle is a geodesic, the edges of the triangle are geodesics, and two geodesics in $\mathbb{H}^2$ meet at most once. Since the path already crossed $p_i$ on entering, it cannot cross $p_i$ again to return. 

When a child window is removed, its virtual start is strictly dominated near the relevant vertex by the occupying virtual start; by Lemma~\ref{lem:connected}, the dominated region is a connected interval near the vertex, so the undominated part of any valid path is always preserved, hence we never remove a window of a valid path. Hence at every vertex and at every level we never remove a window of a valid path. Thus, at the end of the algorithm, the sequence tree contains a window covering every point on every edge that $\pi$ crosses.

By Lemmas~\ref{lem:fan-happy} and~\ref{lem:portal-happy}, $\pi$ crosses each fan edge at most once and each portal side at most twice, so $\pi$ passes through at most $3k$ edges. Each level of the sequence tree advances a window across one edge, so a path crossing at most $3k$ edges is captured within $3k$ levels. Hence $d^* \leq \ell(\pi) = d_S(\alpha, \omega)$.
\end{proof}

\begin{lemma}
    Algorithm~\ref{alg:sd} terminates in $O(g^2)$ time.
\end{lemma}

The fan triangulation has $O(g)$ edges, $O(g)$ vertices, and $O(g)$ triangles, since the number of sides $k$ of $\partial D_{\tilde c}$ satisfies $4g \le k \le 12g-6$ and so $k = O(g)$. Each level of the sequence tree advances a window across one edge, so a path crossing at most $3k$ edges is captured within $3k$ levels. The tree has height at most $3k = O(g)$. By the argument of \cite{chen1996shortest}, the total number of leaves at each level of the sequence tree is $O(g)$. At level $0$ there are $O(g)$ leaves. At each subsequent level, every leaf produces at most two children, but whenever a window forks into two children by taking over occupancy of a vertex, the previous occupant loses one child. Since there are at most $O(g)$ vertices, there are at most $O(g)$ such takeovers per level, each removing one existing child. Hence the leaf count does not grow between levels after it has propagated through all the triangles and remains $O(g)$ throughout. Since the tree has height at most $3k = O(g)$, the total number of windows created is $O(g^2)$. Each window is processed in $O(1)$ time by a constant number of distance comparisons against the angle occupant table. Hence the total running time is $O(g^2)$.

\begin{remark}\label{rem:multi-dest}
We can modify the algorithm to handle multiple endpoints from a single start. Given $h$ destinations, we mark each in its triangle; whenever a window is created over a triangle, we update the current best distance of each destination in that triangle in $O(1)$ time per destination. Since each triangle is visited $O(g)$ times during propagation, each destination is checked $O(g)$ times, for a total of $O(g^2 + gh)$ instead of $h$ separate executions.
\end{remark}

\section{Handling Multiple Shortest Distance Queries}\label{sec:datastructure}
\subsection{Query Data Structure}\label{sec:queryDS}
While the wavefront algorithm we employ is adapted from Chen \& Han \cite{chen1996shortest}, their original algorithm also supports shortest path queries in $O(\log n)$ time after $O(n \log n)$ postprocessing after running their algorithm, via the star unfolding and the Voronoi diagram of source images. This query method implicitly relies on a special property of convex polyhedra: the vertices of the polyhedron define the ridge tree (cut locus) of the source point. Specifically, Sharir \& Schorr \cite{sharir1984shortest} prove that the cut locus of any source point on a convex polyhedron is a tree whose leaves are exactly the vertices of the polyhedron (Lemma 4.5(a) of \cite{sharir1984shortest}), and that shortest paths cannot pass through vertices (Lemma 4.1 of \cite{sharir1984shortest}). These results rely on the fact that curvature is concentrated at the corners of a polyhedral surface. Aronov \& O'Rourke \cite{aronov1991nonoverlap} then use these properties to prove that the star unfolding does not self-overlap and that the ridge tree is precisely the Voronoi diagram of the $O(n)$ source images.

On a smooth hyperbolic surface $S$ curvature is uniform everywhere so neither of these lemmas applies. Shortest paths on $S$ can pass through any point, and the cut locus of $\alpha$ can have leaves anywhere on $S$ with no forced termination points. It remains an open problem whether a hyperbolic analog of the star unfolding exists as a direct consequence of the sequence tree and whether Chen \& Han's query method can be generalized to this setting.

As such, we construct our query data structure by mirroring the process of the discrete geodesic problem but on our sequence tree \cite{mitchell1987discrete}. For each edge of the fan triangulation we maintain an ordered list of windows. By Lemma~\ref{lem:connected}, the region where any virtual start dominates on an edge is a connected interval, so these windows partition each edge into an ordered sequence of non-overlapping intervals each owned by a unique virtual start. We process paths on the sequence tree in a priority queue via the minimum distances from windows onto edges. This requires $O(\log g)$ time per node to pop. Whenever we propagate a virtual start $\tilde{\alpha}'$ onto an edge $e$, we compare its distance function $d_{\mathbb{H}^2}(\tilde{\alpha}', \cdot)$ against the existing windows on $e$ using binary search on the ordered window list, finding the connected interval where $\tilde{\alpha}'$ dominates in $O(\log  g)$ time by Lemma~\ref{lem:connected}. If a virtual start cannot propagate onto the edge by this process we prune the entire branch. Now only virtual starts that contribute to realizable paths are left. By Lemma~\ref{lem:fan-happy} and Lemma~\ref{lem:portal-happy}, each triangle of the fan triangulation is visited by at most $O(g)$ paths of the sequence tree, giving $O(g)$ virtual starts per triangle across $O(g)$ triangles. In each we build a hyperbolic Voronoi diagram of the virtual starts in $O(g\log g)$ time restricted to the triangle~\cite{bogdanov2014}. This allows us to do an angular binary search to determine the triangle to check of the fan triangulation and then a standard Voronoi lookup in that triangle for a runtime of $O(\log g)$ for querying points. Summing over all levels, the priority queue costs $O(\log g)$ per node across $O(g^2)$ nodes, and the per-triangle Voronoi diagrams cost $O(g \log g)$ each across $O(g)$ triangles, giving $O(g^2 \log g)$ preprocessing.

\begin{theorem}[Point queries]\label{thm:query}
After $O(g^2 \log g)$ preprocessing, the shortest path distance $d_S(\alpha, \omega)$ for any query point $\omega \in S$ can be computed in $O(\log g)$ time.
\end{theorem}

\subsection{Recentering a Dirichlet Domain}
Keeping the same structure as the previous subsection yields the farthest distance from $\alpha$ to any point of the surface, $F(\alpha)$. Within each triangle $T$, the restricted Voronoi diagram of the virtual starts partitions $T$ into cells such that in the cell the distance to $\alpha$ is exactly the distance from its virtual start. Since $d_{\mathbb{H}^2}(\tilde{\alpha}'_{\tilde v}, \cdot)$ is convex, each cell maximum is attained at one of its vertices, so $F(\alpha) = \max_{T} \; \max_{\tilde v \in \mathcal{V}(T)} d_{\HY^2}(\tilde{\alpha}'_{\tilde v}, \tilde v)$, where $\mathcal{V}(T)$ is the vertex set of the Voronoi diagram restricted to $T$ and $\tilde{\alpha}'_{\tilde v}$ is the virtual start owning $\tilde v$. Each triangle contributes $O(g)$ Voronoi vertices, so checking distances at all of them takes $O(g^2)$ time, which is dominated by the $O(g^2 \log g)$ preprocessing of Theorem~\ref{thm:query}. As such, we get the following corollary of Theorem~\ref{thm:query}:

\begin{corollary}[Farthest point]\label{cor:farthest}
For any $\alpha \in S$, the farthest distance $F(\alpha)$ and a point realizing it can be computed in $O(g^2 \log g)$ time.
\end{corollary}

In fact, the per-triangle Vorono\"i diagrams contain considerably more information than is required for answering distance queries. Each Vorono\"i cell is associated with a virtual source $\gamma\tilde{\alpha}$, and every point of the cell lies in the translated Dirichlet domain $\gamma D_{\tilde{\alpha}}$. Consequently, unfolding these diagrams into $\mathbb{H}^2$ reconstructs the Dirichlet domain centered at $\tilde{\alpha}$.

\begin{lemma}[Recentering a Dirichlet domain]
\label{lem:dirichlet-from-voronoi}
Let $\alpha\in S$ with lift $\tilde{\alpha}\in D_{\tilde c}$. After computing the data structure query of Section~\ref{sec:queryDS} from $\tilde{\alpha}$, unfolding the per-triangle Vorono\"i diagrams of the virtual sources into $\mathbb{H}^2$ yields the vertices of the Dirichlet domain $D_{\tilde{\alpha}}$.
\end{lemma}

\begin{proof}
Consider a vertex $\tilde v$ of one of the per-triangle Vorono\"i diagrams. By construction, $\tilde v$ is equidistant from at least three virtual sources. Unfolding the corresponding sequence of portal crossings from $\tilde \alpha$ maps $\tilde v$ to $\tilde v$ and at least two of its translates.

It remains to show that no other translate of $\tilde v$ is closer to $\tilde{\alpha}$. Suppose, for contradiction, that there exists $\tilde v'\in\Gamma \tilde v$ such that $d_{\mathbb{H}^2}(\tilde{\alpha},\tilde v')<
d_{\mathbb{H}^2}(\tilde{\alpha},\tilde v)$.
Folding the geodesic from $\tilde{\alpha}$ to $\tilde v'$ back into $D_{\tilde c}$, Remark~\ref{rem:virtual-start-distance} yields a path from $\tilde{\alpha}$ to the representative of $\tilde v$ in $D_{\tilde c}$ whose length is
$d_{\mathbb{H}^2}(\tilde{\alpha},\tilde v')$. This path is strictly shorter than the distance computed by the algorithm, contradicting its correctness. Therefore, the unfolded image of $\tilde v$ lies in the Dirichlet domain $D_{\tilde{\alpha}}$.

Conversely, every vertex of $D_{\tilde{\alpha}}$ arises in this way. Indeed, since $D_{\tilde c}$ is a fundamental domain, it contains a representative of every vertex of $D_{\tilde{\alpha}}$, and the algorithm computes shortest paths from $\tilde{\alpha}$ to every point of $D_{\tilde c}$. Hence, unfolding the per-triangle Vorono\"i diagrams recovers exactly the vertices of $D_{\tilde{\alpha}}$.
\end{proof}

This lemma allows us to design an algorithm that, given a Dirichlet domain $D_{\tilde c}$ and a query point $\tilde{\alpha}$, recenters the Dirichlet domain around $\tilde{\alpha}$.

\begin{algorithm}[!ht]
  \caption{Recentering a Dirichlet domain}
  \label{alg:recenter}
  \KwIn{A hyperbolic surface $S$ given by a Dirichlet domain $D_{\tilde c}$, and a query point $\alpha \in S$ with lift $\tilde\alpha \in D_{\tilde c}$.}
  \KwOut{The Dirichlet domain $D_{\tilde\alpha}$.}

  \textbf{Run the wavefront.} Compute the query data structure of Section~\ref{sec:queryDS} from $\tilde\alpha$, so that every edge of the fan triangulation receives its final set of windows.

  \textbf{Build the per-triangle Vorono\"i diagrams.} For each triangle $T$ of the fan triangulation, collect the $O(g)$ virtual starts whose windows lie on $T$ and build their hyperbolic Voronoi diagram restricted to $T$.

  \textbf{Unfold.} For each Vorono\"i vertex $\tilde v$ reached from a virtual start by crossing a sequence of portals $p_1, \dots, p_j$, apply the side-pairing isometries $\gamma_{p_1} \circ \cdots \circ \gamma_{p_j}$ to $\tilde v$.

  \Return the hyperbolic convex hull of these unfolded vertices.
\end{algorithm}

\begin{remark}
    At Step 2, we have to consider the vertices of the original triangles as vertices of the Vorono\"i even if they will generically not appear in the final Dirichlet domain. 
\end{remark}

We can now prove the complexity result.

\begin{theorem}[Complexity of Algorithm~\ref{alg:recenter}]
\label{thm:recenter}
Let $D_{\tilde c}$ be a Dirichlet domain and let $\alpha\in S$ with lift $\tilde{\alpha}\in D_{\tilde c}$. The Dirichlet domain $D_{\tilde{\alpha}}$ can be computed in $O(g^2\log g)$ time.
\end{theorem}

\begin{proof}
By Lemma~\ref{lem:dirichlet-from-voronoi}, unfolding the per-triangle Voronoi diagrams produced by the data structure of Section~\ref{sec:queryDS} recovers the vertices of $D_{\tilde{\alpha}}$. These diagrams are constructed in $O(g^2\log g)$ time by Theorem~\ref{thm:query}. Since $4g\leq k\leq 12g-6$, the Dirichlet domain has $O(g)$ vertices, and its hyperbolic convex hull can therefore be computed in $O(g\log g)$ time. The overall running time is thus $O(g^2\log g)$.
\end{proof}

\section{Approximate Diameter Computation}\label{sec:diameter}
\subsection{Thick surfaces}
The data structure developed in Section~\ref{sec:datastructure} can also be used to approximate the diameter of a hyperbolic surface. Indeed, Corollary~\ref{cor:farthest} allows us to compute, for any source point $\alpha$, a farthest point from $\alpha$. Evaluating this query for every point of an $\eps$-net immediately yields an additive $\eps$-approximation of the diameter.

\begin{lemma}\label{lem:thickdiam}
Let $S$ be a hyperbolic surface and let $N_\eps$ be an $\eps$-net on $S$. Then there exists a point $p\in N_\eps$ admitting a point $p'\in S$ such that $d(p,p')\ge\operatorname{diam}(S)-\eps$.
\end{lemma}

\begin{proof}
Let $a,b\in S$ satisfy $d(a,b)=\operatorname{diam}(S)$. Since $N_\eps$ is an $\eps$-net, there exists a point $p\in N_\eps$ with $d(a,p)\le\eps$. The triangular inequality yields $d(p,b)\ge d(a,b)-d(a,p)\ge\operatorname{diam}(S)-\eps$, proving the claim by taking $p'=b$.
\end{proof}

At this point, one should be careful when analyzing the complexity of algorithms based on $\eps$-nets. Although the approximation guarantee of Lemma~\ref{lem:thickdiam} depends only on $\eps$, the size of an $\eps$-net may become arbitrarily large as the systole tends to zero because of the long embedded collars surrounding short geodesics. Consequently, the complexity shall not be bounded solely as a function of $g$ and $\eps$ without further assumptions.
There are two ways to overcome this difficulty. One may either assume that the surface is $\eps$-thick, in which case every $\eps$-net contains $O(g/\eps^2)$ points, or replace $\eps$-nets by pseudo-$\eps$-nets, which avoid the thin collars while preserving the approximation guarantee. We first present the simpler algorithm for thick surfaces before extending it to arbitrary hyperbolic surfaces.

\begin{theorem}\label{thm:thickdiam}
Let $S$ be an $\eps$-thick hyperbolic surface and let $N_\eps$ be an $\eps$-net on $S$. Then an additive $\eps$-approximation $\Lambda$ of $\operatorname{diam}(S)$ satisfying $\operatorname{diam}(S)-\eps\le \Lambda\le\operatorname{diam}(S)$ can be computed in time $O(g^3\log g\cdot 1/\eps^2)$.
\end{theorem}

\begin{proof}
For every point $p\in N_\eps$, we compute a farthest point from $p$ using the algorithm of Corollary~\ref{cor:farthest}. By Lemma~\ref{lem:thickdiam}, the maximum of the resulting distances is an additive $\eps$-approximation of the diameter. Each farthest-point query requires $O(g^2\log g)$ time. Since $S$ is $\eps$-thick, $N_\eps$ contains at most $16(g-1)/\eps^2$ points~\cite{despreComputingEnetClosed2024}. Therefore, the overall running time is $O((g/\eps^2)\cdot g^2\log g)=O(g^3\log g\cdot 1/\eps^2)$.
\end{proof}

\subsection{Using Pseudo-$\eps$-Nets on Hyperbolic Surfaces}
We introduce an alternative algorithm based on pseudo $\eps$-nets, whose size depends only on $\eps$ (and the genus), independently of the systole. The key ingredient is a decomposition of $S$ into its $\eps$-thin and $\eps$-thick parts. Before describing the algorithm, we recall the necessary definitions and establish several supporting lemmas.

\begin{definition}[Surface $\eps$-thin and $\eps$-thick decomposition]
    Given $\eps > 0$, the \emph{$\eps$-thin part} of a hyperbolic surface $S$ is:
        $S^{\text{Thin}}_\eps = \{x \in S : \text{inj}(x) < \eps/2\}$,
    where $\text{inj}(x)$ denotes the injectivity radius of $x$ in $S$ (the largest $r$ such that the ball $B(x, r) \subset S$ lifts isometrically to a disk in $\HY^2$). The \emph{$\eps$-thick part} is its complement $S^{\text{Thick}}_\eps = S - S^{\text{Thin}}_\eps$.
\end{definition}

\begin{definition}[Pseudo $\eps$-net]
    Given a hyperbolic surface $S$ and $\eps > 0$, a \emph{pseudo $\eps$-net} of $S$ is a finite point-set $P \subset S^{\text{Thick}}_\eps$ such that $P$ is an $\eps$-net of $S^{\text{Thick}}_\eps$.
\end{definition}

For our choice of $\eps<2\arcsinh(1)$ the \emph{$\eps$-thin part} $S^{\mathrm{thin}}_\eps$ decomposes as a disjoint union of at most $3g - 3$ open \emph{$\eps$-collars} $\mathcal{C}(\gamma_i, \eps)$, one for each simple closed geodesic $\gamma_i$ on $S$ of length at most $\eps$ \cite[Theorem~4.1.6]{buserGeometrySpectraCompact2010}. Each collar is a hyperbolic cylinder around $\gamma_i$. Our algorithm will come in two parts. We handle two cases: if at least one point of the diameter pair lies in 
$S^{\text{Thick}}_\eps$, we find it using the pseudo $\eps$-net; otherwise, 
both points lie in the thin part, and we show they must lie on two distinct 
cylinders. 

 Given this, we only need to pairwise compare the cylinders. As mentioned before, there are only $O(g)$ cylinders. To compare the distances between the $O(g^2)$ pairs of cylinders we will need to look at $\eps$-nets on their boundaries. The cylinders are bounded by pointed geodesics in a pseudo $\eps$-net that have length smaller than $2\arcsinh(1)$ by the Collar Lemma~\cite[Theorem~4.1.6]{buserGeometrySpectraCompact2010}. Since the boundaries are of length $\eps$, putting any $\eps$-net on them involves $O(1)$ points, independent of $g$. We now show that we can snap distances to the net to approximate distances inside the cylinders.

\begin{lemma}[Cylinder boundary snapping] 

\label{lem:cylinder-boundary-snapping} 

Consider two distinct cylinders A, B on S. Let $N_A \subset \bd A$ and $N_B \subset \bd B$ be $\eps$-nets of $\bd A$ and $\bd B$. Then, for all $p \in A$ and $q \in B$ there exist $b_A \in N_A$ and $b_B \in N_B$ such that: $ d_S(p,q) \;\le\; d_A(p, b_A) + d_S(b_A, b_B) + d_B(b_B, q) \;\le\; d_S(p,q) + \eps. $ \end{lemma}

\begin{proof}
    This proof is omitted in this version of the paper.
\end{proof}

We now show that inter-cylinder distances can be computed using the boundary points on each cylinder. Each cylinder has two boundary curves, contributing one point each, so $N_A$ and $N_B$ each consist of two points. On the cylinder, the additively weighted Voronoi diagram of two sites has a single vertex, which is the point we compute.

\begin{lemma}[Same-cylinder diameter pairs]
\label{lem:same-cylinder-snapping}
Let $C$ be a cylinder of $S^{\mathrm{thin}}_\eps$ and let $N_C \subset \bd C$ be an $\eps$-net of $\bd C$. Then for all $p, q \in C$ there exist $c_1, c_2 \in N_C$ such that:
\[
 d_S(p,q) \;\le\; d_S(c_1,c_2) + \eps.
\]
\end{lemma}

\begin{proof}
Proof deferred to the full version of this paper.
\end{proof}

\begin{lemma}[Voronoi vertices approximate the inter-cylinder diameter]
\label{lem:voronoi-approx}
Let $A, B$ be distinct cylinders of $S^{\mathrm{thin}}_\eps$, with designated boundary point sets $N_A \subset \partial A$ and $N_B \subset \partial B$. For each $a \in N_A$, assign the weight $w_a = \min_{b \in N_B} d_S(a, b)$, and symmetrically for each $b \in N_B$. Let $v$ be the vertex of the additively weighted Voronoi diagram of $(N_A, w)$ in $A$, and $u$ the vertex of $(N_B, w)$ in $B$. Then
\[
\Big| d_S(v, u) - \max_{(p,q) \in A \times B} d_S(p,q) \Big| \le 2\eps.
\]
\end{lemma}
\begin{proof}
    This proof is omitted in this version of the paper.
\end{proof}

Additionally we prove that we can compute a weighted Voronoi of two sites on a cylinder in constant time.

\begin{lemma}[Weighted Voronoi diagrams on cylinders]\label{lem:cylinder-voronoi}
Let $C$ be a cylinder of $S^{\mathrm{thin}}_\eps$ with two weighted sites on $\partial C$. The additively weighted Voronoi of these sites inside $C$ can be computed in $O(1)$ time.
\end{lemma}

\begin{proof}
Lift $C$ to a strip $T$ in $\HY^2$, bounded by the two lifted boundary curves of $C$ and by two portal sides identified by the isometry $\tau$. Additively weighted Voronoi cells are star-shaped about their sites, since a point is claimed by a site once its growing ball reaches it, and every point on the geodesic from the site is claimed no later. A bisector winding around $C$ would bound a cell wrapping the core geodesic, which is not star-shaped, so no bisector winds. Each bisector therefore crosses each portal side at most once, and it suffices to consider the lifts $\tau^{-1}\tilde{s}_1, \tilde{s}_1, \tau\tilde{s}_1, \tau^{-1}\tilde{s}_2, \tilde{s}_2, \tau\tilde{s}_2$, compute their weighted hyperbolic Voronoi diagram, and intersect it with $T$. Since this involves only six sites it's constant time.
\end{proof}

In a full version of ``Computing an Epsilon-net of a closed Hyperbolic Surface'' \cite{despre2024computing} they show that these boundary curves are a constant factor of $\eps$ in length, so the whole process only involves a constant amount of sites.

\subsection{Algorithm Description}
To approximate the diameter of $S$, we utilize the pseudo-$\eps$-net $N_\eps$ defined above alongside our multi-end shortest path routine. We rely on the standard thick-thin decomposition of $S$ into a thick part $S^{\text{Thick}}_\eps$ and a thin part $S^{\text{Thin}}_\eps$ consisting of disjoint cylinders. The algorithm proceeds in three main steps:

\begin{algorithm}[!ht]
  \caption{Approximate diameter}
  \label{alg:diameter}
  \KwIn{A hyperbolic surface $S$ given by a Dirichlet domain and the Delaunay triangulation of a pseudo $\eps/2$-net $N_\eps$ on $S^{\mathrm{Thick}}_\eps$.}
  \KwOut{An additive $\eps$-approximation of $\operatorname{diam}(S)$.}

  \textbf{Thick part.} \For{each $p \in N_\eps$}{
    Compute the query data structure of Section~\ref{sec:queryDS} from $p$ and extract the farthest distance $F(p)$ by Corollary~\ref{cor:farthest}. This guarantees we find a value within $\eps$ of the true diameter, provided at least one point of the diameter pair lies in $S^{\text{Thick}}_\eps$.
  }
  Set $D_{\mathrm{thick}} = \max_{p \in N_\eps} F(p)$. 

  \textbf{Thin part.} \For{each pair of distinct cylinders $C_1, C_2$ of $S^{\mathrm{thin}}_\eps$}{
    Let $N_{C_1}, N_{C_2}$ be their boundary points. Assign each site the weight $w_s = \min_t d_S(s, t)$ over the other cylinder's boundary points.

    Compute the additively weighted Voronoi vertices $v \in C_1$ and $u \in C_2$ (Lemma~\ref{lem:cylinder-voronoi}), and record $d_S(v, u)$.
  }
  \For{each cylinder $C$ of $S^{\mathrm{thin}}_\eps$}{
    Let $N_C = \{c_1, c_2\}$ be its boundary points and record $ d_S(c_1,c_2)$.
  }
  Set $D_{\mathrm{thin}} = \max d_S(v, u)$ over all cylinder pairs, and $ d_S(c_1,c_2)$ over all cylinders.

  \Return $\max(D_{\mathrm{thick}}, D_{\mathrm{thin}})$.
\end{algorithm}

\subsection{Correctness and Complexity}

We formalize the correctness of our algorithm by proving that the returned value strictly bounds the error to $\eps$.

\begin{theorem}[Approximation Guarantee]\label{thm:approx-correctness}
    Let $\Lambda$ be the maximum distance value returned by the algorithm. Then $\Lambda$ is an $\eps$-approximation of the true diameter of $S$, satisfying: $\operatorname{diam}(S)-\eps\le \Lambda\le\operatorname{diam}(S)$
\end{theorem}

\begin{proof}
    Let $(x, y)$ be the diameter pair on $S$.

    \textbf{Case 1: At least one point is in the thick part.}
    Without loss of generality, assume $x \in S^{\text{Thick}}_\eps$. By the definition of the pseudo-$\eps$-net $N_\eps$, there exists a point $p \in N_\eps$ such that $d_S(x, p) \leq \eps$. By the triangle inequality, the distance from $p$ to $y$ is bounded by: \[ d_S(p, y) \geq d_S(x, y) - d_S(x, p) \geq \text{diam}(S) - \eps. \]
    In Step 2, the algorithm computes the exact farthest distance from every point $p \in N_\eps$ across the entire surface $S$. Therefore, it will discover a valid distance path from $p$ to another point. Hence the largest possible value returned from this process is at most $\text{diam(S)}$. Thus $| \text{diam}(S) - \Lambda | \leq \eps$.
    
\textbf{Case 2: Both points are in the thin part.}
Suppose $x, y \in S^{\mathrm{thin}}_\eps$. By Lemma~\ref{lem:same-cylinder-snapping}, if $x$ and $y$ lie in the same cylinder then there exist $c_1, c_2 \in N_C$ such that
\[
d_S(x,y) \;\le\; d_S(c_1,c_2) + \eps,
\]
so comparing pairs of boundary points on $\bd C$ recovers $\operatorname{diam}(S) = d_S(x,y)$ to within an additive $\eps$. Otherwise, $x \in A$ and $y \in B$ for distinct cylinders $A, B$. In Step~3, the algorithm assigns the boundary points of $A$ and $B$ their weights and computes $d_S(v, u)$ for the weighted Voronoi vertices $v \in A$, $u \in B$. Since the diameter pair lies on $A$ and $B$, we have $\max_{(p,q) \in A \times B} d_S(p,q) = d_S(x,y) = \operatorname{diam}(S)$. Applying Lemma~\ref{lem:voronoi-approx} with the $\eps/2$-net gives
\[
|d_S(v,u) - \operatorname{diam}(S)| \le 2 \cdot \tfrac{\eps}{2} = \eps.
\]
In all cases the algorithm returns a value $\Lambda$ with $|\operatorname{diam}(S) - \Lambda| \le \eps$.\qedhere

\end{proof}

We are now able to prove the complexity of this algorithm which appear to be the same as the thick surface one. However, notice that it requires an $\eps/2$-net instead of an $\eps$-net.

\begin{theorem}[Time Complexity]\label{thm:approx-runtime}
    Given a parameter $\eps > 0$, the approximation algorithm computes $\Lambda$ in total time $O(g^3 \log g\cdot1/\eps^2)$.
\end{theorem}

\begin{proof}
    The overall time complexity is determined by the summation of the costs of the two algorithmic steps:

    \textbf{Step 1 ($S^{\text{Thick}}_\eps$ Distances):}
    Theorem~\ref{thm:thickdiam} gives the $O(g^3 \log g\cdot1/\eps^2)$.

    \textbf{Step 2 ($S^{\text{Thin}}_\eps$ Extraction):}
    Each boundary of the cylinders have a single vertex by construction. Therefore, constructing the $(\frac{\eps}{2})$-nets on all cylinder boundaries yields a total of $O(g)$ boundary sites across $S$. 
    \begin{itemize}
        \item Running the shortest path algorithm from each of these boundary sites takes $O(g \cdot g^2 \log g)$.
        \item There are $O(g^2)$ pairs of cylinders. Since we already have computed the datastructure associated to each point of the boundary of the cylinders in Step 1, for each pair, querying the distances between their respective boundary sites requires an $O(\log g)$ tree search per point pair for a total complexity of $O(g^2 \log g)$.
        \item Computing the weighted Voronoi diagram for each pair of cylinders involves only a constant number of sites, taking $O(1)$ time per pair. Across all $O(g^2)$ pairs, this takes $O(g^2)$ time.
    \end{itemize}
    
    Summing the construction, thick-part routing, and thin-part extraction bounds yields the stated total time complexity of $O(g^3 \log g\cdot1/\eps^2)$.
\end{proof}

\section{Exact Diameter Computation}\label{sec:exact-diameter}

\subsection{Geometric Bounds and Farthest Points}

Throughout this section we assume the surface is given by a finite set of generators of $\Gamma$, each a hyperbolic isometry. Let $\mathbb{K}$ be the rational extension containing the real and imaginary parts of the coefficients of the generators of $\Gamma$. Any isometry in  $\Gamma$ thus has the real and imaginary parts of its coefficients in $\mathbb{K}$. Note that in this section matrices representing isometries are not normalized to $\det=1$ to avoid square roots. We show that $\cosh(\operatorname{diam}(S))$ is algebraic over $\mathbb{K}$, and give an explicit (if exponential-time) algorithm computing $\operatorname{diam}(S)$. This resolves a major open problem in hyperbolic computational geometry. We begin with the definition of the \emph{diameter} of a hyperbolic surface.

\begin{definition}[Diameter]
    The diameter of a hyperbolic surface is the longest distance path on $S$:
    \[
\text{diam}(S) = \sup_{s_1, s_2 \in S} d_S(s_1, s_2).
\]
\end{definition}
 Since $S$ is compact, the supremum of this distance function is attained. We give an exact algorithm for computing $\text{diam}(S)$ and show moreover that $\cosh(\operatorname{diam}(S))$ lies in a finite extension of $\mathbb{K}$, so that the diameter is determined by algebraic data over the field of the input. Let us call any pair $(s_1, s_2)$ achieving the supremum a \emph{diameter pair}. For a fixed point $s \in S$, we define its \emph{farthest distance} as $F(s) = \max_{s' \in S} d_S(s, s')$, so that $\text{diam}(S) = \max_{s \in S} F(s)$.

\begin{lemma}[Farthest point at a vertex]\label{lem:farthest-vertex}
    For any $s \in S$, the farthest point from $s$ on $S$ is the projection onto $S$ of a vertex of $D_{\tilde s}$, for any lift $\tilde s$.
\end{lemma}
\begin{proof}
    The farthest point from a site in its Voronoi cell given a convex distance function is always a vertex of that cell. Since $D_{\tilde{s}}$ is the Voronoi cell of $\tilde{s}$ in the diagram of the orbit $\Gamma\tilde{s}$, and the surface distance $d_S(s, s') = \min_{\gamma \in \Gamma} d_{\mathbb{H}^2}(\tilde{s}, \gamma\tilde{s}')$, the farthest point on $S$ from $s$ corresponds to the farthest vertex of $D_{\tilde{s}}$.
\end{proof}

Now we prove that only a finite number of isometries of $\Gamma$ are relevant to computing $F(s)$ for any $s \in D_{\tilde{s}_0}$.

\begin{lemma}[Relevant isometries]\label{lem:relevant-isometries}
    Fix a reference lift $\tilde{s}_0 \in \mathbb{H}^2$ and define $\Gamma_{4}(\tilde{s}_0) = \{ \gamma \in \Gamma \setminus \{e\} \mid d(\tilde{s}_0, \gamma \tilde{s}_0) \leq 4\,\text{diam}(S)\}$. Then for every $s \in S$ with lift $\tilde s \in D_{\tilde s_0}$, any vertex of $D_{\tilde s}$ achieving $F(s)$ is the intersection of two bisectors, between $\tilde s$ and orbit images $\gamma_1 \tilde s$, $\gamma_2 \tilde s$ respectively, with $\gamma_1, \gamma_2 \in \Gamma_{4}(\tilde s_0)$. Moreover, since $\Gamma$ acts by isometries, the bound
    \[
    |\Gamma_{4}(\tilde{s}_0)| \leq \frac{\cosh(5\,\text{diam}(S))}{2(g-1)}
    \]
    is independent of the choice of $\tilde{s}_0$.
\end{lemma}
\begin{proof}
    Let $\tilde{s} \in D_{\tilde{s}_0}$ and let $v$ be a vertex of $D_{\tilde{s}}$ achieving $F(s)$, given as the intersection of the bisectors between $\tilde{s}$ and the orbit images $\gamma_1 \tilde{s}$, $\gamma_2 \tilde{s}$, where $\gamma_1, \gamma_2 \in \Gamma$. We bound $d(\tilde{s}_0, \gamma_i \tilde{s}_0)$ for $i \in \{1,2\}$ to show $\gamma_i \in \Gamma_{4}(\tilde{s}_0)$. By the triangle inequality:
    \[
    d(\tilde{s}_0, \gamma_i \tilde{s}_0) \leq d(\tilde{s}_0, \tilde{s}) + d(\tilde{s}, \gamma_i \tilde{s}) + d(\gamma_i \tilde{s}, \gamma_i \tilde{s}_0) \leq \text{diam}(S) + 2\,\text{diam}(S) + \text{diam}(S) = 4\,\text{diam}(S),
    \]
    where $d(\tilde{s}_0, \tilde{s}) \leq \text{diam}(S)$ since $\tilde{s} \in D_{\tilde{s}_0}$, $d(\tilde{s}, \gamma_i \tilde{s}) \leq 2\,\text{diam}(S)$ since $v$ is equidistant from $\tilde{s}$ and $\gamma_i \tilde{s}$ and $d(\tilde{s}, v) \leq \text{diam}(S)$, and $d(\gamma_i \tilde{s}, \gamma_i \tilde{s}_0) \leq \text{diam}(S)$ since $\gamma_i$ is an isometry and $\tilde{s} \in D_{\tilde{s}_0}$. Hence $\gamma_i \in \Gamma_{4}(\tilde{s}_0)$. 

    The bound on $|\Gamma_{4}(\tilde{s}_0)|$ follows from a packing argument: the orbit points $\{\gamma \tilde{s}_0\}_{\gamma \in \Gamma_{4}}$ lie in a ball of radius $4\,\text{diam}(S)$, each with a disjoint Dirichlet cell of area $4\pi(g-1)$. Since cells may only be partially contained in the ball, we use a ball of radius $5\,\text{diam}(S)$ to ensure full containment, giving:
    \[
    |\Gamma_{4}(\tilde{s}_0)| \leq \frac{2\pi(\cosh(5\,\text{diam}(S)) - 1)}{4\pi(g-1)} 
    \leq \frac{\cosh(5\,\text{diam}(S))}{2(g-1)}.
    \qedhere\]
\end{proof}

\subsection{The Wall Arrangement}

Fix $s_0 \in S$ with lift $\tilde{s}_0$. As the basepoint $\tilde{s}$ varies over $D_{\tilde{s}_0}$, the combinatorial structure of the Voronoi diagram of the orbit $\Gamma \tilde{s}$ changes. A combinatorial transition occurs exactly when four orbit points become cocyclic, that is, when a Voronoi vertex switches adjacency. We call the locus of such transitions a \emph{wall}.

\begin{definition}[Wall]
For three distinct non-identity isometries $\gamma_1, \gamma_2, \gamma_3 \in \Gamma_{4}$, the \emph{wall} $W_{\gamma_1, \gamma_2, \gamma_3}$ is the set of basepoints $\tilde{s} \in \mathbb{H}^2$ such that there exists $v \in \mathbb{H}^2$ with
\[
d(v, \tilde{s}) = d(v, \gamma_1 \tilde{s}) = d(v, \gamma_2 \tilde{s}) = d(v, \gamma_3 \tilde{s}) \quad \text{and} \quad d(v, \tilde{s}) \leq d(v, \gamma \tilde{s}) \text{ for all } \gamma \in \Gamma.
\]
\end{definition}

\begin{lemma}[Algebraic degree of walls]\label{lem:wall-degree}
    Each wall $W_{\gamma_1, \gamma_2, \gamma_3}$ is a real algebraic curve of degree at most $8$ in the coordinates $\tilde{s}=(x_1, x_2)$ where $\tilde{s} = x_1 + i x_2$, with coefficients in $\mathbb{K}$.
\end{lemma}

\begin{proof}
    Since hyperbolic balls are Euclidean balls in the Poincar\'e disk model, four points are cocyclic if and only if their cross-ratio is real in the complex plane. Hence the walls are defined as: 
    \begin{equation}\label{eq:wall-curve}
        \operatorname{Im}\left( \frac{(\tilde{s} - \gamma_2(\tilde{s}))(\gamma_1(\tilde{s}) - 
        \gamma_3(\tilde{s}))}{(\tilde{s} - \gamma_3(\tilde{s}))(\gamma_1(\tilde{s}) - \gamma_2(\tilde{s}))} 
        \right) = 0.
    \end{equation}
    Let $\Delta_j(\tilde{s}) = \bar{b}_j \tilde{s} + \bar{a}_j$ be the denominator of 
    $\gamma_j(\tilde{s}) = \frac{a_j \tilde{s} + b_j}{\bar{b}_j \tilde{s} + \bar{a}_j}$. 
    We make two simplifications:
    \[
        \tilde{s} - \gamma_j(\tilde{s}) = \frac{\bar{b}_j \tilde{s}^2 + (\bar{a}_j - a_j)\tilde{s} - 
        b_j}{\Delta_j(\tilde{s})} = \frac{M_j(\tilde{s})}{\Delta_j(\tilde{s})},
    \]
    where $M_j(\tilde{s})$ is degree $2$ in $\tilde{s}$. Second,
    \[
        \gamma_j(\tilde{s}) - \gamma_k(\tilde{s}) = \frac{(a_j\bar{b}_k - a_k\bar{b}_j)\tilde{s}^2 + 
(a_j\bar{a}_k + b_j\bar{b}_k - a_k\bar{a}_j - b_k\bar{b}_j)\tilde{s} + 
(b_j\bar{a}_k - b_k\bar{a}_j)}{\Delta_j(\tilde{s})
        \Delta_k(\tilde{s})}=\frac{M_{j,k}(\tilde{s})}{\Delta_j(\tilde{s})
        \Delta_k(\tilde{s})},
    \]
    where $M_{j,k}(\tilde{s})$ is also degree $2$ in $\tilde{s}$. Substituting into 
    \eqref{eq:wall-curve}, the denominators $\Delta_1(\tilde{s})\Delta_2(\tilde{s})
    \Delta_3(\tilde{s})$ cancel, giving us:
    \[
        \operatorname{Im}\left(\frac{M_2(\tilde{s}) M_{1,3}(\tilde{s})}{M_3(\tilde{s}) 
        M_{1,2}(\tilde{s})}\right) = 0.
    \]
    Denoting the numerator and denominator as $P_4(\tilde{s})$ and $Q_4(\tilde{s})$ 
    respectively, both quartic in $\tilde{s}$ because they are products of quadratics. The condition 
    $\operatorname{Im}(P_4/Q_4) = 0$ is equivalent to $\frac{P_4(\tilde{s})/Q_4(\tilde{s})-\overline{P_4(\tilde{s})/Q_4(\tilde{s})}}{2i}=0$ which is equivalent to $P_4(\tilde{s})\overline{Q_4({\tilde{s}})} - \overline{P_4({\tilde{s}})}Q_4(\tilde{s}) = 0$, 
    which is a polynomial of degree $8$ in $\tilde{s}=(x_1, x_2)$. Note that the coefficients of $M_j$ and $M_{j,k}$, $P_4$, and $Q_4$, are polynomials in $a_j, b_j$ and their conjugates, this polynomial thus has coefficients in $\mathbb{K}$.
\end{proof}

\begin{lemma}[Arrangement complexity]\label{lem:arrangement}
    Let $N = |\Gamma_{4}|$ denote the number of relevant isometries. The walls $\{W_{\gamma_1, \gamma_2, \gamma_3}\}$ over all triples from $\Gamma_{4}$ form a planar arrangement with:
    \begin{enumerate}
        \item At most $\binom{N}{3} = O(N^3)$ curves,
        \item At most $O(N^6)$ intersection points, 
        \item $O(N^6)$ cells.
        \item Cells that are semi-algebraic sets defined over $\mathbb{K}$.
    \end{enumerate}
\end{lemma}

\begin{proof}
There are $\binom{N}{3}$ triples from $\Gamma_{4}$, giving $O(N^3)$ walls. Since two algebraic curves of degree $8$ intersect in at most $64$ points, the total number of intersection points is a constant multiple of the amount of pairs, $O(N^6),$ and so the number of cells in the arrangement is $O(N^6)$. Since the wall polynomials have coefficients in $\mathbb{K}$ by Lemma~\ref{lem:wall-degree}, each cell of the arrangement is a semi-algebraic set defined over $\mathbb{K}$.
\end{proof}

In each cell $\mathcal{C}$ of this arrangement the combinatorial structure of the Voronoi diagram of $\Gamma \tilde{s}$ is constant, for any $\tilde{s}$ in $\mathcal{C}$. In particular, the set of isometries of $\Gamma$ defining the vertices and edges of the Voronoi diagram of $\Gamma \tilde s$ is invariant, let these be $\Gamma_{\mathcal{C}}$. We denote $v_{\gamma_1, \gamma_2}(\tilde{s})$ to be the Voronoi vertex defined by the orbit points $\tilde{s}$, $\gamma_1(\tilde{s})$, and $\gamma_2(\tilde{s})$, i.e. the center of the hyperbolic circle with points $\tilde{s}$, $\gamma_1(\tilde{s})$, and $\gamma_2(\tilde{s})$ on its boundary. It is now sufficient to study the function $\max_{\gamma_1,\gamma_2\in\Gamma_{\mathcal{C}}}d(\tilde{s},v_{\gamma_1,\gamma_2}(\tilde{s}))$ over the wall arrangement.

\subsection{Semi-Algebraic Optimization}

\begin{lemma}
\label{lem:functionToMaxmizise}
The value $\displaystyle\max_{\tilde s \in {\mathcal{C}},\, \gamma_1, \gamma_2 \in \Gamma_{\mathcal{C}}} \cosh{d(\tilde s, v_{\gamma_1, \gamma_2}(\tilde s))}$ for a cell ${\mathcal{C}}$ of the wall arrangement is algebraic over $\mathbb{K}$.
\end{lemma}

\begin{proof}
    Since
    $\operatorname{arccosh}$ is monotone, maximizing $d(\tilde{s}, v_{\gamma_1,\gamma_2}(\tilde{s}))$ is equivalent to maximizing \[\frac{2|\tilde{s} - v_{\gamma_1,\gamma_2}(\tilde{s})|^2}{(1-|\tilde{s}|^2)(1-|v_{\gamma_1,\gamma_2}(\tilde{s})|^2)}.\] Each $\gamma_i(\tilde{s})$ is a rational function of $\tilde{s}$ with coefficients in $\mathbb{K}$. Since hyperbolic circles are Euclidean circles in the Poincar\'e disk, the circle through $\tilde{s}$, $\gamma_1(\tilde{s})$, and $\gamma_2(\tilde{s})$ has Euclidean center and radius rational in these points, and its hyperbolic center is obtained by solving a quadratic over $\mathbb{K}(\tilde{s})$. Hence $v_{\gamma_1, \gamma_2}(\tilde{s})$ is algebraic of degree at most $2$ over $\mathbb{K}(\tilde{s})$, so the objective is algebraic over $\mathbb{K}$. Together with Lemma~\ref{lem:arrangement}, which gives that the cell is semi-algebraic over $\mathbb{K}$, the optimization problem is semi-algebraic over $\mathbb{K}$.
\end{proof}

\subsection{Algorithm and Runtime}

\begin{theorem}[Exact diameter computation]\label{thm:exact-diameter}
    The diameter of a compact hyperbolic surface $S$ of genus $g \geq 2$ is exactly computable in time $O(N^6 (\log N + g^2))$ where $N = |\Gamma_{4}| \leq \frac{\cosh(5\,\mathrm{diam}(S))}{2(g-1)}$.
    In particular, the runtime is $O\!\left(\frac{e^{30\,\mathrm{diam}(S)}}{(g-1)^6}\cdot \mathrm{diam}(S)\right)$.
\end{theorem}

\begin{proof}
    By Lemma~\ref{lem:farthest-vertex}, the diameter is realized at a vertex of $D_{\tilde{x}}$ for some $x \in S$. By Lemma~\ref{lem:relevant-isometries}, only isometries in $\Gamma_{4}$ are relevant. We construct the wall arrangement of Lemma~\ref{lem:arrangement}, giving $O(N^6)$ cells such that within each cell the combinatorial structure of the Voronoi diagram is constant. At each wall $W_{\gamma_1, \gamma_2, \gamma_3}$ we store its defining triple $(\gamma_1, \gamma_2, \gamma_3)$, the isometries whose orbit points $\tilde{s}, \gamma_1\tilde{s}, \gamma_2\tilde{s}, \gamma_3\tilde{s}$ are cocyclic along that wall. Then we sweep through the cells, maintaining the current set of Voronoi vertices by inserting and removing pairs $(\gamma_1, \gamma_2)$ at each wall crossing in $O(\log N)$ time. Within each cell, we optimize the semi-algebraic function of Lemma~\ref{lem:functionToMaxmizise} over the current set of vertices. For each cell this is $O(g^2)$ algebraic systems to solve that have constant complexity (number of variables and degree), so all together $O(N^6 (\log N + g^2))$. The global maximum over all cells gives $\mathrm{diam}(S)$ exactly.
\end{proof}

\begin{theorem}[Algebraicity of the diameter]\label{thm:algebraic-diameter}
Let $S$ be a compact hyperbolic surface of genus $g \geq 2$ given by generators of $\Gamma$. Let $\mathbb{K}$ be the rational extension containing the real and imaginary parts of the coefficients of the generators of $\Gamma$. Then $\cosh(\operatorname{diam}(S))$ is algebraic over $\mathbb{K}$.
\end{theorem}
\begin{proof}
By Lemma~\ref{lem:farthest-vertex} and Lemma~\ref{lem:relevant-isometries}, $\operatorname{diam}(S)$ is realized as $d(\tilde s, v_{\gamma_1,\gamma_2}(\tilde s))$ for some cell of the wall arrangement, some $\tilde s$ in that cell, and some $\gamma_1, \gamma_2 \in \Gamma_4$. By Lemma~\ref{lem:arrangement} the cell is semi-algebraic over $\mathbb{K}$, and by Lemma~\ref{lem:functionToMaxmizise} the optimization over it is semi-algebraic over $\mathbb{K}$. Hence the optimal $\tilde s^\ast$ is algebraic over $\mathbb{K}$, and $v_{\gamma_1,\gamma_2}(\tilde s^\ast)$ is algebraic over $\mathbb{K}(\tilde s^\ast)$, so algebraic over $\mathbb{K}$. The objective
\[
\cosh(\operatorname{diam}(S)) = 1 + \frac{2|\tilde s^\ast - v(\tilde s^\ast)|^2}{(1-|\tilde s^\ast|^2)(1-|v(\tilde s^\ast)|^2)}
\]
is a rational expression in quantities algebraic over $\mathbb{K}$, hence algebraic over $\mathbb{K}$.
\end{proof}

\bibliography{shortcuts,hilbert}

\begin{thebibliography}{10}

\bibitem{abikoff1981uniformization}
William Abikoff.
\newblock The uniformization theorem.
\newblock {\em The American Mathematical Monthly}, 88(8):574--592, 1981.

\bibitem{aronov1991nonoverlap}
Boris Aronov and Joseph O'rourke.
\newblock Nonoverlap of the star unfolding.
\newblock In {\em Proceedings of the seventh annual symposium on Computational geometry}, pages 105--114, 1991.

\bibitem{bogdanov2014}
Mikhail Bogdanov, Olivier Devillers, and Monique Teillaud.
\newblock {Hyperbolic Delaunay Complexes and Voronoi Diagrams Made Practical}.
\newblock {\em {Journal of Computational Geometry}}, 5(1):56--85, 2014.
\newblock URL: \url{https://inria.hal.science/hal-00961390}, \href {https://doi.org/10.20382/jocg.v5i1a4} {\path{doi:10.20382/jocg.v5i1a4}}.

\bibitem{bogdanov2016delaunay}
Mikhail Bogdanov, Monique Teillaud, and Gert Vegter.
\newblock Delaunay triangulations on orientable surfaces of low genus.
\newblock In {\em 32nd International Symposium on Computational Geometry}, pages 20--1, 2016.

\bibitem{buserGeometrySpectraCompact2010}
Peter Buser.
\newblock {\em Geometry and Spectra of Compact {Riemann} Surfaces}.
\newblock Modern Birkh\"auser Classics. Birkh\"auser Boston, 1st edition, 2010.
\newblock \href {https://doi.org/10.1007/978-0-8176-4992-0} {\path{doi:10.1007/978-0-8176-4992-0}}.

\bibitem{Carroll2019}
Sean~M Carroll.
\newblock {\em Spacetime and geometry}.
\newblock Cambridge University Press, 2019.

\bibitem{chen1996shortest}
Jindong Chen and Yijie Han.
\newblock Shortest paths on a polyhedron, part i: Computing shortest paths.
\newblock {\em International Journal of Computational Geometry \& Applications}, 6(02):127--144, 1996.

\bibitem{despreComputingDirichletDomain2023}
Vincent Despr{\'e}, Benedikt Kolbe, Hugo Parlier, and Monique Teillaud.
\newblock Computing a {Dirichlet} domain for a hyperbolic surface.
\newblock In {\em 39th International Symposium on Computational Geometry ({SoCG})}, volume 258, pages 27:1--27:15, 2023.
\newblock \href {https://doi.org/10.4230/LIPIcs.SoCG.2023.27} {\path{doi:10.4230/LIPIcs.SoCG.2023.27}}.

\bibitem{despre2024representing}
Vincent Despr{\'e}, Benedikt Kolbe, and Monique Teillaud.
\newblock Representing infinite periodic hyperbolic delaunay triangulations using finitely many dirichlet domains.
\newblock {\em Discrete \& Computational Geometry}, 72(1):1--28, 2024.

\bibitem{despre2024computing}
Vincent Despr{\'e}, Camille Lanuel, and Monique Teillaud.
\newblock Computing an epsilon-net of a closed hyperbolic surface.
\newblock 2024.

\bibitem{despreComputingEnetClosed2024}
Vincent Despr{\'e}, Camille Lanuel, and Monique Teillaud.
\newblock Computing an $\varepsilon$-net of a closed hyperbolic surface.
\newblock Preprint, 2024.
\newblock URL: \url{https://hal.science/hal-04466350}.

\bibitem{despreFlippingGeometricTriangulations2020}
Vincent Despr{\'e}, Jean-Marc Schlenker, and Monique Teillaud.
\newblock Flipping geometric triangulations on hyperbolic surfaces.
\newblock In {\em 36th International Symposium on Computational Geometry ({SoCG})}, volume 164, pages 35:1--35:16, June 2020.
\newblock \href {https://doi.org/10.4230/LIPIcs.SoCG.2020.35} {\path{doi:10.4230/LIPIcs.SoCG.2020.35}}.

\bibitem{iordanov2016implementing}
Iordan Iordanov and Monique Teillaud.
\newblock {\em Implementing Delaunay triangulations of the Bolza surface}.
\newblock PhD thesis, INRIA Nancy, 2016.

\bibitem{krioukov2010hyperbolic}
Dmitri Krioukov, Fragkiskos Papadopoulos, Maksim Kitsak, Amin Vahdat, and Mari{\'a}n Bogun{\'a}.
\newblock Hyperbolic geometry of complex networks.
\newblock {\em Physical Review E—Statistical, Nonlinear, and Soft Matter Physics}, 82(3):036106, 2010.

\bibitem{matsumoto2021novel}
Hirotaka Matsumoto, Takahiro Mimori, and Tsukasa Fukunaga.
\newblock Novel metric for hyperbolic phylogenetic tree embeddings.
\newblock {\em Biology Methods and Protocols}, 6(1):bpab006, 2021.

\bibitem{mitchell1987discrete}
Joseph~SB Mitchell, David~M Mount, and Christos~H Papadimitriou.
\newblock The discrete geodesic problem.
\newblock {\em SIAM Journal on Computing}, 16(4):647--668, 1987.

\bibitem{nickel2017poincare}
Maximillian Nickel and Douwe Kiela.
\newblock Poincar{\'e} embeddings for learning hierarchical representations.
\newblock {\em Advances in neural information processing systems}, 30, 2017.

\bibitem{sharir1984shortest}
Micha Sharir~( and Amir Schorr.
\newblock On shortest paths in polyhedral spaces.
\newblock In {\em Proceedings of the sixteenth annual ACM symposium on Theory of computing}, pages 144--153, 1984.

\bibitem{stepanyantsDiameterCompactRiemann2023}
Huck Stepanyants, Alan Beardon, Jeremy Paton, and Dmitri Krioukov.
\newblock Diameter of compact {Riemann} surfaces, 2023.
\newblock URL: \url{https://arxiv.org/abs/2301.10844}.

\end{thebibliography}
\end{document}